\documentclass[article]{IEEEtran123}
\usepackage{amsmath,graphicx}
\usepackage{amsmath}
\usepackage{amsthm}
\usepackage{amssymb}
\usepackage{mathtools}
\usepackage{tikz}
\usepackage{hyperref}
\usepackage{tikzscale}
\usepackage{pgfplots}
\pgfplotsset{compat=1.16}
\usepackage{graphicx}
\usepackage{multirow}
\usepackage{grffile}
\usepackage{float}
\usepackage{subcaption}
\usepackage{siunitx} 
\usepackage{pgfplotstable}
\usepackage{pgfplots}

\newcommand{\R}{\mathbb{R}}
\newcommand{\N}{\mathcal{N}}
\newcommand{\E}{\mathbb{E}}

\newcommand{\zh}{\hat{z}}
\newcommand{\Ec}{\mathcal{E}}

\newtheorem{problem}{Problem}
\newtheorem{remark}{Remark}
\newtheorem{example}{Example}
\newtheorem{theorem}{Theorem}

\renewcommand{\AA}[1]{#1}

\title{Optimal Intermittent Particle Filter}
\author{Antoine ASPEEL, Amaury GOUVERNEUR, Rapha\"el M. JUNGERS and Benoit MACQ\thanks{A.A. is supported by the Walloon Region, under grant RW-DGO6-Biowin-Bidmed. R.J. is a FNRS Research Associate. He is also supported by the Walloon Region, the Innoviris Foundation, and the FNRS. This project has received funding from the European Research Council (ERC) under the European Union’s Horizon 2020 research and innovation programme under grant agreement No 864017 - L2C.}
\thanks{ICTEAM Institute, UCLouvain, Avenue Georges Lemaître 4-6, 1348 Louvain-la-Neuve, Belgium}}
\begin{document}
%
\maketitle

\IEEEpeerreviewmaketitle


\begin{abstract} 
The problem of the optimal allocation (in the expected mean square error sense) of a measurement budget for particle filtering is addressed. We propose three different optimal intermittent filters, \AA{whose optimality criteria depend on the information available at the time of decision making}. For the first, the stochastic program filter, the measurement times are given by \AA{a policy that determines whether a measurement should be taken based on the measurements already acquired.} The second, called the offline filter, determines all measurement times at once by solving a combinatorial optimization program before any measurement acquisition. For the third one, which we call online filter, each time a new measurement is received, the next measurement time is recomputed to take all the information that is then available into account. We prove that in terms of expected mean square error, the stochastic program filter outperforms the online filter, which itself outperforms the offline filter. \AA{However, these filters are generally intractable. For this reason, the filter estimate is approximated by a particle filter. Moreover, the mean square error is approximated using a Monte-Carlo approach, and different optimization algorithms are compared to approximately solve the combinatorial programs (a random trial algorithm, greedy forward and backward algorithms, a simulated annealing algorithm, and a genetic algorithm).} Finally, the performance of the proposed methods is illustrated on two examples: a tumor motion model and a common benchmark for particle filtering.

\end{abstract}

\section{Introduction}
\IEEEPARstart{S}{tochastic} nonlinear dynamical systems have shown their ability to model a number of real-world problems \cite{arnold2012stochastic, adomian1988nonlinear, freund2000stochastic}. Particle filtering is a popular approach to estimate the state of such systems from a set of noisy measurements \cite{arulampalam2002tutorial}. This tool has been largely used, among others, in computer vision \cite{hue2001particle,cho2006real,andreasson2005localization}, positioning and navigation \cite{gustafsson2002particle,gustafsson2010particle}, chemistry \cite{shenoy2011practical,jha2016particle}, mechanics \cite{cadini2009monte}, robotics \cite{thrun2002particle}, and medicine \cite{rahni2011particle}. In practice, performing measurements may be difficult due to energy consumption, economic constraints, or health hazards. For instance, in tumor tracking based on X-ray images, the number of X-ray acquisitions should be minimized in order to limit patients' exposure to harmful radiation \cite{sharp2004prediction}.

Under such constraints, the problem is to select the best times to measure the system. In other words, one has a measurement budget and has to choose when to acquire measurements. The optimality criterion is to minimize the expected filtering \textit{mean square error} (MSE) over the complete time horizon.

Intuitively, there are various reasons why regular measurement times may be suboptimal. For example, if the system is (almost) deterministic but there is a large uncertainty in the initial state, it may be better to take measurements as early as possible. However, if the signal amplitude varies greatly over time, it may be worthwhile to take measurements when the signal-to-noise ratio is greatest.


First, the problem of optimal measurement budget allocation is formalized as a nonlinear multistage \emph{stochastic program}. The measurement times are selected one after the other and each measurement time is selected taking the measurements already observed into account.

Next, an easier related problem that we call the \emph{offline program} is presented. It corresponds to selecting all the measurement times at once and offline, i.e., before any measurement acquisition. In the case of real-time applications where the time between each measurement is short, being able to calculate the measurement times offline can be decisive. However, filtering a signal using the measurement times given by this offline program gives, on average, poorer filtering performance because in that case, the information present in the measurements already acquired is ignored.

Then, we develop an online adaptation of the Offline problem which we call the \emph{online program}. This gives better filtering performance, at the cost of a higher computational cost, as computations must be carried out online.

\subsection{Related work}
In recent decades, the problem of estimating the state of a system subject to missing measurements has been extensively studied in both the linear case (see e.g., \cite{sun2008optimalElsevier, rutledge2020finite, liang2010optimal} and references there in) and the nonlinear case (see e.g., \cite{li2016particle, chen2010nonlinear, yang2019particle} and references there in). The latter use the particle filtering framework. In these previous studies, missing measurements occur randomly. In this paper, on the contrary, we aim to choose the measurement times optimally. The measurement times (or equivalently, the missing measurement times), are, for our case, parameters to be optimized.

The \emph{sensor scheduling} (or sensor management) \emph{problem} is closely related to ours and has been widely studied. It consists in choosing at each time step one (or a few) active sensors among a set in order to minimize the variance of the estimation error \AA{\cite{vitus2012efficient, krishnamurthy2002algorithms, garnett2010bayesian, yang2018adaptive,  singh2003stochastic, singh2007simulation, doucet2002particle, singh2004variance, he2006sensor, li2009approximate}}.

\AA{In some of the works on the sensor scheduling problem \cite{krishnamurthy2002algorithms, yang2018adaptive, he2006sensor, li2009approximate},} a cost is (or can be) associated to the use of each sensor and the objective is to minimize a combination of the error variance and the cost of the sensors. \AA{In this case, we speak of the sensor scheduling problem \emph{with a sensing cost}.}

\AA{The problem studied in this article is a constrained version of the sensor scheduling problem without sensing cost. Indeed, consider the sensor scheduling problem where an active sensor must be selected among two available sensors, one corresponding to a measurement acquisition and the other corresponding to no measurements, i.e., it returns a value independent of the system state. By adding a constraint on the maximum number of uses of the measuring sensor, we find the problem of optimal allocation of a measurement budget.}

\AA{In addition, the problem we address can be reduced to the sensor scheduling problem with a sensing cost in the following way. Consider the sensor scheduling problem with a sensing cost with two available sensors: the first one is equivalent to not measuring (this sensor is costless but returns a value independent of the system state); the second sensor has a cost $c>0$. Then by choosing a sensing cost $c$ such that the obtained solution respects the measurement budget (e.g., by using a binary search on $c$), one can reduce our constrained problem to a sensor scheduling problem with a sensing cost.}



In the linear and Gaussian case, \cite{vitus2012efficient} proposes two algorithms to solve the sensor scheduling problem. The first one provides an optimal solution but can be computationally expensive. The second one provides a suboptimal solution and depends on a tuning parameter to make a trade-off between the quality of the solution and the simplicity of the problem. \AA{In that setting, the estimate is given by the Kalman filter. Using the monotonicity and concavity of the Riccati equation, a condition for the non-optimality of the initialization of a schedule is derived. The exact algorithm uses this condition to prune the search tree of all possible sensor schedules. On the other hand, the suboptimal algorithm uses a relaxed condition, and the optimality gap is proven to be upper bounded by a linear function of the relaxation parameter. Given the central role played by the Riccati equation in the proposed method, this approach seems difficult to generalize to the nonlinear and non-Gaussian case.} For a literature review on the sensor scheduling problem in the linear and Gaussian case, see references in \cite{vitus2012efficient}.


\AA{The sensor scheduling problem with a sensing cost has received attention when the state space is a finite set, i.e., the dynamics is a finite hidden Markov model (HMM) \cite{krishnamurthy2002algorithms}. In that article, the estimate is computed using a HMM state filter and the optimal sensor scheduling policy is obtained by stochastic dynamic programming. They propose two efficient methods to obtain a suboptimal solution: the first one is a one-step look-ahead approximation, the other one is based on Lovejoy's algorithm for partially observable Markov decision processes. The proposed methods rely on the fact that the state space and the measurement space are finite sets, which is not the case in the problem we study.}

In the case of a continuous state space, \cite{garnett2010bayesian} proposes a Gaussian process global optimization method used as a black-box optimizer to dynamically select sensors in order to minimize the error at the end of a \AA{time} horizon. \AA{The approach is model free (the system dynamics is unknown) and only the last measurement is used to estimate the state of the system, the previous measurements being used to choose the last active sensor. On the contrary, in our work, the (known) system dynamics is used to choose when to measure. Another difference is that we want to minimize the average MSE.}

\AA{For nonlinear but Gaussian dynamics, a sparsity promoting approach is proposed in \cite{yang2018adaptive} to choose a set of active sensors at each time step. A trade-off is made between the number of active sensors and the current posterior Cramér-Rao lower bound (PCRLB). This is justified by the fact that the PCRLB is a lower bound on the mean square estimation error. The main differences with our approach are that we tackle the general case of a non-Gaussian dynamics, we do not consider a one-step look-ahead approximation, and we directly minimize the MSE instead of the PCRLB.}

The sensor scheduling problem has also been studied when the set of sensors is continuous (e.g., to optimize the spatial position of a mobile sensor). For this problem, \cite{singh2003stochastic, singh2007simulation} propose an algorithm using a simulation-based gradient approximation. \AA{Since the core of the method is based on the gradient of the objective function with respect to the selected sensor, this approach is not applicable when the set of available sensors is finite.}


\AA{For the sensor scheduling problem in the nonlinear and non-Gaussian case, an optimal one-step look-ahead policy is proposed in \cite{doucet2002particle, singh2004variance} to maximize the Kullbach Leibler divergence between filtering and prediction densities.} Intuitively, it maximizes the information obtained at the current measurement time. \AA{Again, the main difference with this work is that we do not consider the one-step look-ahead approximation, and we minimize the average MSE.}


\AA{In \cite{he2006sensor, li2009approximate}, the sensor scheduling problem with a sensing cost for nonlinear and non-Gaussian dynamics is formalized as a continuous partially observable Markov decision process. The $Q$-function is approximated by the policy rollout method. This requires a base policy, i.e., a heuristic, to choose future actions in order to estimate the cumulated costs that will follow the next action. For example, in \cite{he2006sensor}, the available sensors are spatially distributed to track a moving target and the base policy selects the sensor closest to the estimated position of the target. As the authors acknowledge ``The choice of a base policy may have a significant impact on the performance of the rollout policy.'' On the contrary, our method does not require a priori knowledge of a heuristic to select the measurement times.}

\AA{Table \ref{tab:intro:sota:sensorScheduling} summarizes the state of the art concerning the sensor scheduling problem.}

\begin{table*}
\AA{
\caption{\AA{Summary of the state of the art of the sensor scheduling problem. When the sensors are chosen by a one-step look-ahead policy, the mention ``(one step)'' is indicated in the column ``Objective function''.}}
\label{tab:intro:sota:sensorScheduling}
\centering
\begin{tabular}{|l|l|l||l|}
\hline
Dynamics & Sensing cost & Objective function & References \\
\hline
\hline
linear and Gaussian & no & average MSE & \cite{vitus2012efficient} \\
\hline
hidden Markov model & yes & average MSE & \cite{krishnamurthy2002algorithms}\\
\hline
model-free & no & terminal root MSE & \cite{garnett2010bayesian}\\
\hline
nonlinear and Gaussian & yes & posterior Cram\'er-Rao lower bound (one step) & \cite{yang2018adaptive} \\
\hline
\multirow{3}{*}{nonlinear and non-Gaussian}
            & \multirow{2}{*}{no}
                & average MSE & \cite{singh2003stochastic} \cite{singh2007simulation} \\
                \cline{3-4}
                && Kullback-Leibler divergence between prediction and filtering densities (one-step) & \cite{doucet2002particle} \cite{singh2004variance} \\
            \cline{2-4}
            & yes & average MSE & \cite{he2006sensor} \cite{li2009approximate} \\
\hline
\end{tabular}
}
\end{table*}



In the particular case of linear systems subject to Gaussian noise processes, the selection of optimal measurement times over a finite time horizon has been studied using the Kalman filtering framework in both discrete \cite{aspeel2019optimal} and continuous-time \cite{sano1970measurement, aksenov2019optimal} settings.

\subsection{Contribution}

This paper addresses the problem of optimal allocation of a measurement budget in the discrete-time nonlinear case with disturbance and measurement noise processes following arbitrary distributions.


In \cite{aspeel2020optimal}, we have proposed an initial intermittent filter (which is reported as Problem \ref{prob:offline} here). Our contributions in the present paper are that (i) we show that the problem of optimal intermittent filtering can be modeled as a stochastic program (see Problem \ref{prob:stochProg}); (ii) based on this, we propose a recursive version of optimal intermittent filter (see Problem \ref{prob:online}); (iii) we prove the Theorem \ref{thm:SP<online<offline} that states that the expected mean square filtering error is smaller for the stochastic program filter than for the online filter, which itself is smaller than for the offline filter; (iv) new optimization algorithms are tested (greedy forward, greedy backward, and simulated annealing); and (v) the results are presented (among others) on a new model inspired by real-world tumor tracking applications.

Overall, the two main contributions of this paper are to propose an efficient algorithm to solve the problem of optimal measurement times selection; and to show the interest of intermittent measurements in particle filtering.

\subsection{Paper outline}

The rest of this paper is organized as follows: Section \ref{sec:materianlsAndMethods} presents how to define filtering with intermittent measurements (Subsection \ref{sec:IPF}); different versions of optimal intermittent filters (Subsection \ref{sec:optimalIPF}) and how to compute them numerically with different optimization algorithms (subsections \ref{sec:numericalApprox} and \ref{sec:optimizationAlgo}). Results and discussion are presented in Section \ref{sec:resultsAndDiscussion}, where a model of tumor motion is proposed for benchmarking (Subsection \ref{sec:tumorModel}), in addition to a common benchmark for particle filters (Subsection \ref{sec:toy_ex}), the optimization algorithms are compared (Subsection \ref{sec:compareOptimizationAlgo}), and the filtering performance of the offline and online particle filters are compared with a regular particle filter (Subsection \ref{sec:filteringPerformance}). Finally, Section \ref{sec:conclusion} concludes and discusses possible improvements and perspectives.

A \textsc{Matlab} (MathWorks, Natick, Massachusetts, USA) implementation of the presented algorithms and the code that generate the figures is available on \textsc{GitHub} at\\
\href{https://github.com/AmauryGouverneur/Optimal_Measurement_Times_For_Particle_Filtering}{\textsc{github.com/AmauryGouverneur/Optimal\_\\ Measurement\_Times\_For\_Particle\_Filtering}}.

\section{Materials and methods}\label{sec:materianlsAndMethods}

\subsection{Intermittent filter}\label{sec:IPF}
A discrete time stochastic nonlinear dynamic system describes the evolution of a state $x(t)$ over the finite time horizon $t=0,\dots,T$. One wants to estimate a quantity $z(t)$ related to $x(t)$ and has access to previously acquired noisy measurements $y(t)$ of $x(t)$. Measurements are not available at each time step. More precisely, only $N$ measurements are available, at times $t_1,\dots,t_N\in\{0,\dots, T\}$. This is modelled as
\begin{alignat}{3}
x(t+1) &=&&\ f_t(x(t),w(t))  &\text{\ \ for\ \ } & t=0,\dots,T-1, \label{eq:model:x}\\
y(t_j)   &=&&\ g_{t_j}(x(t_j),v(t_j))  &\text{\ \ for\ \ } & j=1,\dots,N, \label{eq:model:y}\\
z(t)   &=&&\ h_t(x(t)) &\text{\ \ for\ \ } & t=0,\dots,T,\label{eq:model:z}\\ 
x(0)   &\sim&&\ \mathcal{F},\label{eq:model:x0}
\end{alignat}
with $t_i\neq t_j$ if $i\neq j$, and $x(t)\in\R^n$, $y(t)\in\R^m$ and $z(t)\in\R^p$. In addition, $w(t)$ and $v(t)$ are random processes with known probability density functions. Functions $f_t(\cdot,\cdot)$, $g_t(\cdot,\cdot)$ and $h_t(\cdot)$ are known and have compatible dimensions. The initial state $x(0)$ follows a known distribution $\mathcal{F}$. 

For instance, in a tumor tracking problem based on X-ray images, $x(t)\in\R^6$ can be a state vector containing the tumor’s position and velocity in 3-dimensional space, $y(t)\in\R^2$ can be the 2-dimensional projection of the target and $z(t)\in\R^3$ the position of the mass center in 3-dimensional space.


Let us introduce some notations and definitions. We use the so-called Matlab notation, for $j\leq k$, $t_{j:k}\coloneqq \{t_j, t_{j+1},\dots,t_k\}$ and $y(t_{j:k})\coloneqq \{y(t_j),y(t_{j+1}),\dots,y(t_k)\}$.

We define the \emph{intermittent filter} estimate,
\begin{align}\label{eq:def:zh}
\zh(t|y(t_{1:j}))\coloneqq
\E_{x(0),w(0),\dots,w(t-1)}[z(t)|y(t_k),\forall t_k\leq t],
\end{align}
where $\E_X[\cdot]$ holds for the expectation operator according to the random variable $X$. The name “intermittent” emphasizes that measurements are not accessible at each time step. Subsection \ref{sec:numericalApprox:estimate} will explain how an intermittent particle filter can be used to calculate an approximation of the intermittent filter estimate.

Let's define the expected filtering error variance,
\begin{align}
&\Ec[t|y(t_{1:j})]\coloneqq \notag\\
&\E_{x(0),w(0),\dots,w(t-1)}[\|z(t)-\zh(t|y(t_{1:j}))\|^2|y(t_k),\forall t_k\leq t], \label{eq:def:Ec}
\end{align}
where $\|\cdot \|$ is the Euclidean norm. A numerical approach to estimate this quantity is described in Paragraph \ref{sec:MC}.

\begin{remark}\label{remark:notation}
\AA{Thanks to the condition $\forall t_k\leq t$ in definitions (\ref{eq:def:zh}) and (\ref{eq:def:Ec}), it follows that if $t_k\leq t<t_{k+1}\leq t_j$, it holds that $\hat{z}(t|y(t_{1:j})) = \hat{z}(t|y(t_{1:k}),y(t_{k+1:j})) = \hat{z}(t|y(t_{1:k}))$, and $\Ec[t|y(t_{1:j})] = \Ec[t|y(t_{1:k}),y(t_{k+1:j})] = \Ec[t|y(t_{1:k})]$.} This ensures the causality of the filter by avoiding future measurements of influencing the present estimation. Otherwise, these two quantities are random variables because they depend on $y(t_1),\dots,y(t_j)$ and therefore on $v(t_1),\dots,v(t_j)$, which are random.
\end{remark}

\begin{remark}\label{remark:estimator}
Instead of the estimator of $z(t)$ defined in \eqref{eq:def:zh}, one could use any estimator that can be defined from the distribution of $x(t)|(y(t_k),\forall t_k\leq t)$, for example, the \emph{maximum a posteriori estimator} could replace \eqref{eq:def:zh} using the method from \cite{saha2009particle}. This flexibility is allowed because we will use a particle filter which estimates the whole distribution of $x(t)|(y(t_k),\forall t_k\leq t)$.

Moreover, since we are going to use black-box optimization algorithms, another quality criterion could be used instead of the error variance \eqref{eq:def:Ec}, for example, the squared Euclidean norm could be replaced by the $1-$norm or the $\infty-$norm.
\end{remark}

\subsection{Optimal intermittent filters}\label{sec:optimalIPF}

An \emph{optimal intermittent filter} is an intermittent filter whose measurement times have been chosen according to a certain optimality criterion. In this section, we present three different optimality criteria for selecting the measurement times. Each optimality criterion induces \AA{a different} optimal intermittent filter.

\subsubsection{Stochastic program filter}
In the stochastic program filter, each measurement time $t_{j+1}$ is chosen taking the already acquired measurements $y(t_{1:j})$ into account. Then, each measurement time $t_{j+1}$ is the solution of an optimization program.

More formally, the \emph{stochastic program filter} is the intermittent filter whose measurement times are the solution of the following stochastic program.

\begin{problem}[Stochastic Program]\label{prob:stochProg}
\begin{align*}
V_0(\emptyset;\emptyset)=&\min_{t_1}\Bigg\{ \sum_{t=0}^{t_1-1}\Ec[t|\emptyset]\\
&+\E_{y(t_1)}\Bigg[\min_{t_2}\Bigg\{\sum_{t=t_1}^{t_2-1}\Ec[t|y(t_1)] + \cdots\\
&+\E_{y(t_{N-1})}\Bigg[\min_{t_N}\Bigg\{\sum_{t=t_{N-1}}^{t_N-1}\Ec[t|y(t_{1:N-1})]\\
&+\E_{y(t_N)}\Bigg[\sum_{t=t_N}^T\Ec[t|y(t_{1:N})] \Bigg]\Bigg\}\Bigg]\cdots\Bigg\}\Bigg]\Bigg\},
\end{align*}
subject to $0\leq t_1< t_2<\dots<t_N\leq T$.
\end{problem}

Introducing the notation $t_0\coloneqq0$, the problem can be stated recursively as,
\begin{align}
&V_j(t_{1:j};y(t_{1:j}))=\min_{t_{j+1}}\Bigg\{ \sum_{t=t_j}^{t_{j+1}-1}\Ec[t|y(t_{1:j})]\nonumber\\
&+ \E_{y(t_{j+1})}\left[ V_{j+1}(t_{1:j+1};y(t_{1:j+1})) \right] \ \text{with\ } t_j<t_{j+1}\leq T \Bigg\},\label{eq:SP:recursion}
\end{align}
with the terminal condition,
\begin{align}
V_{N}(t_{1:N};y(t_{1:N}))=\sum_{t=t_N}^{T}\Ec[t|y(t_{1:N})].\label{eq:SP:terminalCond}
\end{align}

Solving this multistage stochastic program requires finding an optimal policy for each $t_{j+1}$. Such optimal policies are functions that associate the next optimal measurement time with the previously acquired measurements, i.e., it is $(y(t_1),\dots,y(t_j))\mapsto t_{j+1}$.

Under certain restrictive assumptions, including linearity and finite stochastic outcomes, stochastic multistage programs can be solved using scenario trees and duality \cite{shapiro2014lectures}. For more complicated problems such as ours (nonlinear, continuous random variables), one approach is to use reinforcement learning. In this paper, we propose another approach, which is to optimize the variables all in once and offline, i.e., before any measurement acquisition. We use this approximation and we call it the offline approach.

\subsubsection{Offline filter}
The offline filter is optimal (in the sense of the expected MSE) if all the measurement times must be chosen before receiving any measurement. Each measurement time $t_j$ is independent of the measurements that preceded it, $y(t_{1:j-1})$.

Formally, the \emph{offline filter} is the intermittent filter whose measurement times are the solution of the following offline program.

\begin{problem}[Offline Program]\label{prob:offline}
\begin{align*}
J_0(\emptyset;\emptyset)\coloneqq\min_{0\leq t_1<\cdots<t_N\leq T} \E_{y(t_{1:N})}\Bigg[\sum_{t=0}^T\Ec[t|y(t_{1:N})]\Bigg].
\end{align*}
\end{problem}
Because each optimal measurement time $t_{j}$ is independent of the previously acquired measurements $y(t_{1:j-1})$, all these optimal measurement times $t_{1:N}$ can be computed at once, before any measurement acquisition. This can be a decisive advantage for certain real-world applications if the time between two discrete time steps is too short to solve an optimization program. Unlike Problem \ref{prob:stochProg}, here the optimization variables are no longer functions but natural numbers. 

\subsubsection{Online filter}
To enhance filtering performance, once the measurements $y(t_{1:j})$ have been acquired, we may want to use them to recompute the next measurement time $t_{j+1}$ including the information previously acquired, i.e., the previous measurements. This requires solving $N$ different optimization programs, one for each measurement time.

Formally, the \emph{online filter} is the intermittent filter for which each measurement time $t_{j+1}$ (for $j=0,\dots,N-1$) is the solution of the following $(j+1)^\text{th}$ online program:

\begin{problem}[$(j+1)^\text{th}$ Online Program]\label{prob:online}
When measurements $y(t_{1:j})$ have been acquired, the next measurement time $t_{j+1}$ is computed by solving
\begin{align}
&J_j(t_{1:j};y(t_{1:j}))\coloneqq \min_{t_{j+1}<\cdots<t_N\leq T} \E_{y(t_{j+1:N})}\Bigg[\nonumber\\
&\sum_{t=t_{j}}^{T}\Ec[t|y(t_{1:N})] \Bigg], \text{\ such\ that\ } t_{j+1}>t_j,\label{eq:def:J}
\end{align}
where $t_0\coloneqq 0$ but the constraint is replaced by $t_1\geq 0$ when $j=0$.
\end{problem}
In this $(j+1)^\text{th}$ program, only the measurement time $t_{j+1}$ is used by the online filter. To avoid ambiguities, it will sometimes be noted $t_{j+1}^o$. The other variables of the $(j+1)^\text{th}$ program, i.e., the $t_{j+2:N}$, are not the measurement times of the online filter because they will be replaced by the solutions of the next programs. Such online programs can be seen as a recursive version of the offline program.

\begin{remark}\label{remark:t1}
Setting $j=0$ in Problem \ref{prob:online} gives exactly Problem \ref{prob:offline}. Consequently, the optimal $t_1$s are the same for these two problems.
\end{remark}

\begin{remark}\label{remark:kalman}
If the dynamical system (\ref{eq:model:x})-(\ref{eq:model:x0}) is linear and Gaussian, i.e., functions $f_t(\cdot,\cdot), g_t(\cdot,\cdot)$ and $h_t(\cdot)$ are linear and distributions of $w(t), v(t)$ and $x(0)$ are Gaussian, then for given measurement times $t_1,\dots,t_N$, the optimal filtered estimate $\zh(t|y(t_{1:j}))$ is given explicitly by the Kalman filter \cite[Theorem 2]{kalman1960new}.

In addition, the variance of the filtering error $\Ec[t|y(t_{1:j})]$ is independent of the measurements $y(t_{1:j})$, which implies that the Kalman gains can be computed offline (see, e.g., \cite[p.3]{aspeel2019optimal} for details). Consequently, the expectation operators in Problems \ref{prob:stochProg}, \ref{prob:offline} and \ref{prob:online} are equivalent to identity operators, which gives the same problem three times. In conclusion, in the linear Gaussian case, the optimal solutions of Problems \ref{prob:stochProg}, \ref{prob:offline} and \ref{prob:online} are the same.
\end{remark}

If $t_{1:j}$ are the $j$ first measurement times and $y(t_{1:j})$ the corresponding measurements, we denote $F_j(t_{1:j};y(t_{1:j}))$ the remaining cost-to-go from time $t_j$ using the online filter. For $j=0,\dots,N-1$, it is
\begin{align*}
&F_j(t_{1:j};y(t_{1:j}))=\sum_{t=t_j}^{t_{j+1}^o-1}\Ec[t|y(t_{1:j})]\\
&+\E_{y(t_{j+1}^o)}\Bigg[ \sum_{t=t_{j+1}^o}^{t_{j+2}^o-1} \Ec[t|y(t_{1:j}),y(t_{j+1}^o)] + \cdots\\
&+\E_{y(t_{N-1}^o)}\Bigg[ \sum_{t=t_{N-1}^o}^{t_N^o-1} \Ec[t|y(t_{1:j}),y(t_{j+1:N-1}^o)]\\
&+\E_{y(t_N^o)}\Bigg[ \sum_{t=t_{N}^o}^{T} \Ec[t|y(t_{1:j}),y(t_{j+1:N}^o)] \Bigg]\Bigg]\dots\Bigg],
\end{align*}
and for $j=N$, it is
\begin{align}
F_{N}(t_{1:N};y(t_{1:N}))=\sum_{t=t_N}^{T}\Ec[t|y(t_{1:N})].\label{eq:online:terminalCond}
\end{align}

With this definition, the following recursive relation holds,
\begin{align}
&F_j(t_{1:j};y(t_{1:j}))= \sum_{t=t_j}^{t_{j+1}^o-1}\Ec[t|y(t_{1:j})]\nonumber\\
&+ \E_{y(t_{j+1}^o)} \Bigg[ F_{j+1}(t_{1:j},t_{j+1}^o;y(t_{1:j}),y(t_{j+1}^o)) \Bigg].\label{eq:online:recursion}
\end{align}

In the following theorem, we establish that the expected mean square filtering error is smaller or equal for the stochastic program filter than for the online filter, which itself is smaller or equal than for the offline filter.

\begin{theorem}\label{thm:SP<online<offline}
$V_0(\emptyset;\emptyset) \leq F_0(\emptyset;\emptyset) \leq J_0(\emptyset;\emptyset)$.
\end{theorem}

\begin{proof}
The proof is available in the appendix.
\end{proof}

The difference between Problems \ref{prob:stochProg} and \ref{prob:online} may not seem obvious. Indeed, they both require computing the measurement times online by solving $N$ optimization problems. To clarify the difference, we propose an example where these two problems give different solutions.

\begin{example}
Consider three time steps, i.e., $T=2$, with $N=2$ measurements and consider the system illustrated in Figure \ref{fig:counterExample} and summarized as
\begin{align}\label{eq:counterExample:sys:first}
y(0)=z(0)=x(0) &=\begin{cases} 1, & \text{with\ probability\ 1/2}\\
        -1, & \text{with\ probability\ 1/2}, \end{cases}\\
y(1)=z(1)=x(1) &= \begin{cases} w(0), & \text{if}\ x(0)=1 \\
    0, & \text{if}\ x(0)=-1, \end{cases}\\
y(2)=z(2)=x(2) &= \begin{cases} 0, & \text{if}\ x(1)\neq 0 \\
    w(1), & \text{if}\ x(1)=0, \end{cases}\label{eq:counterExample:sys:last}
\end{align}
where $w(0)$ and $w(1)$ independently follow a uniform distribution between $-6$ and $6$, i.e., $w(0),w(1)\sim\mathcal{U}([-6,6])$. 

\begin{figure}
\centering
\begin{subfigure}{0.45\columnwidth}
\begin{tikzpicture}[scale=1]
\draw [->,>=latex] (0,-3.5) -- (0,3.5);
\draw (0,3.5) node[above]{$x(t),\ y(t),\ z(t)$};
\draw [->,>=latex] (0,0) -- (2.5,0);
\draw (2.5,0) node[right]{$t$};
\foreach \y in {-6,...,6} {
  \draw (0.15,\y/2) -- (-0.15,\y/2) node[left] {$\y$};
}
\draw (1,0.15) -- (1,-0.15) node[below left] {$1$};
\draw (2,0.15) -- (2,-0.15) node[below] {$2$};
\draw (0,0.5) node{$\bullet$};
\draw (2,0) node{$\bullet$};
\draw (1,-3) -- (1,3);
\draw (0.85,-3) -- (1.15,-3);
\draw (0.85,3) -- (1.15,3);
\end{tikzpicture}
\caption{If $x(0)=1$.}
\label{fig:counterExample:+1}
\end{subfigure}
\hfill
\begin{subfigure}{0.45\columnwidth}
\begin{tikzpicture}[scale=1]
\draw [->,>=latex] (0,-3.5) -- (0,3.5);
\draw (0,3.5) node[above]{$x(t),\ y(t),\ z(t)$};
\draw [->,>=latex] (0,0) -- (2.5,0);
\draw (2.5,0) node[right]{$t$};
\foreach \y in {-6,...,6} {
  \draw (0.15,\y/2) -- (-0.15,\y/2) node[left] {$\y$};
}
\draw (1,0.15) -- (1,-0.15) node[below] {$1$};
\draw (2,0.15) -- (2,-0.15) node[below left] {$2$};
\draw (0,-0.5) node{$\bullet$};
\draw (1,0) node{$\bullet$};
\draw (2,-3) -- (2,3);
\draw (1.85,-3) -- (2.15,-3);
\draw (1.85,3) -- (2.15,3);
\end{tikzpicture}
\caption{If $x(0)=-1$.}
\label{fig:counterExample:-1}
\end{subfigure}
\caption{Illustration of the system (\ref{eq:counterExample:sys:first}) through (\ref{eq:counterExample:sys:last}). Vertical lines represent uniform distributions $\mathcal{U}([-6,6])$. Each of scenarios (a) and (b) has a $50/50$ chance of occurring, depending on whether $x(0)=1$ or $-1$.}
\label{fig:counterExample}
\end{figure}
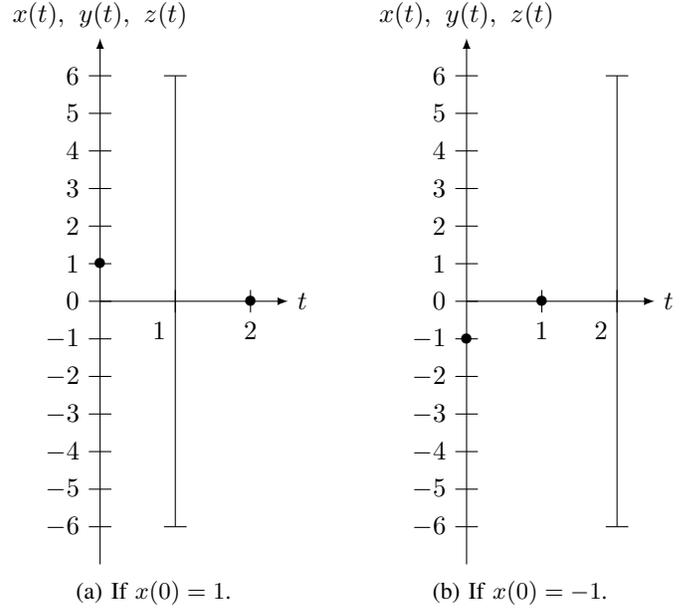

First observe that it is possible to make no filtering error. To do this, set the first measurement time to $t_1=0$ and then, if $y(0)=1$, set the second to $t_2=1$ and on the contrary, if $y(0)=-1$, set the second to $t_2=2$. With this schedule, a measurement is acquired whenever there is an uncertainty, which guarantees that there is no filtering error. It is the solution of Problem \ref{prob:stochProg}.

However, we will show that the solutions of both Problems \ref{prob:offline} and \ref{prob:online} are to choose $t_1=1$ and $t_2=2$ which results in a filtering error at time $t=0$.

Doing an exhaustive search for Problem \ref{prob:offline} shows that choosing $t_1=1$ and $t_2=2$ is optimal. This is due to the fact that if there is no measurement at $t=1$ or at $t=2$, a large error will be made half the time because of the uniformly distributed $x(t)$. Then, because the optimal $t_1$s for Problems \ref{prob:online} and \ref{prob:offline} are the same (see Remark \ref{remark:t1}), $t_1=1$ is also optimal for the former. But then, the only remaining possibility for $t_2$ is $t_2=2$. This shows that the optimal solutions of both Problems \ref{prob:offline} and \ref{prob:online} are $t_1=1$ and $t_2=2$ which gives a filtering error at time $t=0$.

This example illustrates that the optimal time $t_1$ of Problem \ref{prob:online} is computed by assuming that the time $t_2$ cannot depend on $y(t_1)$.
\end{example}

\subsection{Numerical approximations}\label{sec:numericalApprox}
In this section, we present the numerical approximations which are used to make Problems \ref{prob:offline} and \ref{prob:online} tractable. Firstly, we briefly explain how to approximate numerically the estimate $\zh(t|y(t_{1:j}))$ using a particle filter and then how to approximate the objective functions of these two problems. We reserve the term ``intermittent filter" for the ideal filter and use ``intermittent particle filter" for its approximation. Let us emphasize that Theorem \ref{thm:SP<online<offline} applies only in the ideal case.

\subsubsection{Approximation of the estimate}\label{sec:numericalApprox:estimate}
The estimate $\zh(t|y(t_{1:j}))$ is not straightforward to compute but can be approximated using a particle filter.

Essentially, a particle filter algorithm alternates between (i) a \textit{prediction step} (also called mutation), used to compute the estimate at the next time step from the estimate at the current step; and (ii) a \textit{correction step} (also called selection) that updates the current estimation to incorporate the information acquired in the last measurement. To deal with intermittent measurements, the correction step (ii) is skipped when no measurement is available, i.e., when $t\neq t_k$ for all $k$. \AA{For reasons of numerical stability, the particle weights are computed and stored in the logarithmic domain (see, e.g., \cite[Section 4.3.1]{lien2012comparison} for details).}


In this paper, we use the \emph{sampling importance resampling particle filter} (see \cite[Algorithm 4]{arulampalam2002tutorial}).

In general, the approximation provided by the particle filter is biased when the number of particles is finite (the bias being $O(1/\text{\#\ particles})$) \cite[Section 3.2]{doucet2009tutorial}. However, we neglect this aspect in this article and leave the study of the impact of this bias for future work.

In what follows, we will make a slight abuse of notation by using the same notation $\zh(t|y(t_{1:j}))$ for the quantity (\ref{eq:def:zh}) and for its approximation calculated by the particle filter.

\subsubsection{Approximation of the objective functions}\label{sec:MC}
The objective functions of Problems \ref{prob:offline} and \ref{prob:online} contain an expectation operator, which makes these functions challenging to evaluate. To tackle this problem, we use a Monte Carlo approximation of the expectation.

Because the objective function of Problem \ref{prob:offline} is the same as the objective function of Problem \ref{prob:online} for $j=0$ (see Remark \ref{remark:t1}), we focus on the objective function of Problem \ref{prob:online} for any $j$, which includes Problem \ref{prob:offline}.

The Monte Carlo approximation of the objective function of Problem \ref{prob:online} is represented in Figure \ref{fig:schemeMC}. Assuming that the first $j$ measurements $y(t_{1:j})$ have already been acquired and that future measurement times $t_{j+1:N}$ are fixed, one can simulate $K$ realizations of the state $\{x^k(t)\}^{k=1,\dots,K}_{t=0,\dots,T}$ drawn according to Equations (\ref{eq:model:x}) and (\ref{eq:model:x0}). The corresponding quantities that remain to be estimated $\{z^k(t)\}^{k=1,\dots,K}_{t=t_j,\dots,T}$ can be computed using Equation (\ref{eq:model:z}). Similarly, corresponding measurements $\{y^k(t_l)\}^{k=1,\dots,K}_{l=j+1,\dots,N}$ can be simulated thanks to Equation (\ref{eq:model:y}). For each $k$, the simulated measurement sequence can be concatenated to the known measurements to give a complete measurement sequence $y(t_1),\dots,y(t_j),y^k(t_{j+1}),\dots,y
^k(t_N)$, from which the particle filter computes the estimates $\{\zh^k(t|y(t_1),\dots,y^k(t_N))\}_{t=t_j,\dots,T}$ of $\{z^k(t)\}_{t=t_j,\dots,T}$. Then, thanks to Definition (\ref{eq:def:Ec}), the Monte Carlo estimator is computed as
\begin{align}\nonumber
&\E_{y(t_{j+1:N})}\Bigg[\sum_{t=t_{j}}^{T}\Ec[t|y(t_{1:N})] \Bigg]= \E_{y(t_{j+1:N})}\Bigg[ \nonumber\\
&\sum_{t=t_{j}}^{T} \E_{x(0),w(0),\dots,w(T-1)}\Big[ \nonumber\\
&\|z(t)-\zh(t|y(t_{1:N}))\|^2\Big|y(t_l),\forall t_l\leq t\Big]\Bigg] \nonumber\\
&\approx \dfrac{1}{K}\sum_{k=1}^K\sum_{t=t_j}^T \|z^k(t)-\zh^k(t|y(t_{1:j}),y^k(t_{j+1:N}))\|^2. \label{eq:MC:estimation}
\end{align}

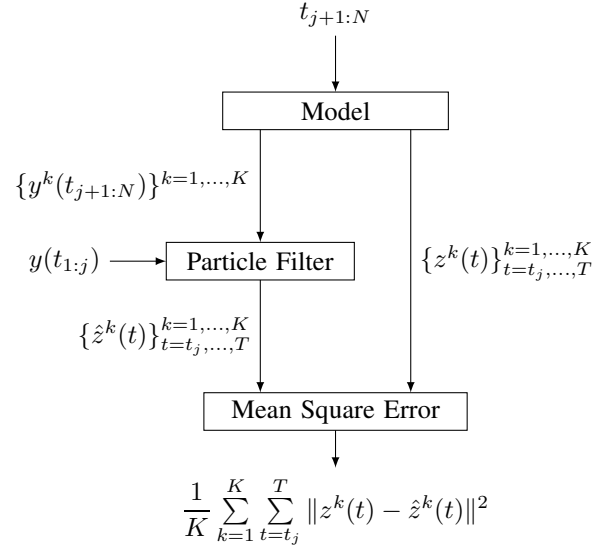
\begin{figure}[t]
\centering
\begin{tikzpicture}
\draw (0,3) node[above]{$t_{j+1:N}$};
\draw [->,>=latex] (0,3) -- (0,2.25);
\draw (0,2) node{Model};
\draw (-1.5,1.75) rectangle (1.5,2.25);
\draw [->,>=latex] (-1,1.75) -- (-1,0.25);
\draw (-1,0) node{Particle Filter};
\draw (-2.25,-0.25) rectangle (0.25,0.25);
\draw [->,>=latex] (-3,0) -- (-2.25,0);
\draw (-3,0) node[left]{$y(t_{1:j})$};
\draw (-1,1) node[left]{$\{y^k(t_{j+1:N})\}^{k=1,\dots,K}$};
\draw [->,>=latex] (-1,-0.25) -- (-1,-1.75);
\draw (-1.75,-2.25) rectangle (1.75,-1.75);
\draw (-1,-1) node[left]{$\{\zh^k(t)\}^{k=1,\dots,K}_{t=t_j,\dots,T}$};
\draw (0,-2) node{Mean Square Error};
\draw [->,>=latex] (1,1.75) -- (1,-1.75);
\draw (1,0) node[right]{$\{z^k(t)\}^{k=1,\dots,K}_{t=t_j,\dots,T}$};
\draw [->,>=latex] (0,-2.25) -- (0,-2.75);
\draw (0,-2.75) node[below]{$\dfrac{1}{K}\sum\limits_{k=1}^K\sum\limits_{t=t_j}^T \|z^k(t)-\zh^k(t)\|^2$};
\end{tikzpicture}
\caption{Representation of the Monte Carlo algorithm that estimates $\E_{y(t_{j+1:N})}[\sum_{t=t_{j}}^{T}\Ec[t|y(t_{1:N})] ]$. The outputs generated by the Model block are drawn according to Equations (\ref{eq:model:x}) through (\ref{eq:model:x0}) and the distributions of $v(t)$ and $w(t)$. The particle filter computes an estimate $\zh^k(t)=\zh^k(t|y(t_{1:j}),y^k(t_{j+1:N}))$ of $z^k(t)$ from previous intermittent measurements. The Mean Square Error block computes the mean squared of the difference between its inputs; it is the right-hand side of Equation (\ref{eq:MC:estimation}).}
\label{fig:schemeMC}
\end{figure}

Computing this Monte Carlo approximation requires running the particle filter $K$ times, which can be computationally expensive.

The fact that we have access only to an approximation of the objective function requires particular attention during the optimization. Note that since the particle filter provides a biased approximation, the approximation of the cost function may be biased as well. However, accounting for this bias is outside the scope of this paper and is left to future work.

\subsection{Optimization algorithms}\label{sec:optimizationAlgo}
Solving the combinatorial optimization Problem \ref{prob:offline} (or equivalently, in light of Remark \ref{remark:t1}, Problem \ref{prob:online} with $j=0$) with an exhaustive search would require evaluating the objective function $\frac{(T+1)!}{(T+1-N)!N!}$ times, i.e., the number of subsets of size $N$ in a set of size $T+1$. This is computationally intractable for large $N$ and $T$ values.

In this section, we present five different heuristic optimization algorithms to find an approximate solution for Problems \ref{prob:offline} and \ref{prob:online} efficiently.

\subsubsection{Random trial algorithm (RT)}
The random trial algorithm samples measurement time sets uniformly at random and evaluates their corresponding costs. The measurement time set with minimum cost is returned.

\subsubsection{Greedy forward algorithm (GF)}
The greedy forward algorithm begins with an empty measurement set. Sequentially, it adds the measurement times one at a time such that, at each iteration, the added measurement time minimizes the cost. It stops when the set contains $N$ measurement times.


\subsubsection{Greedy backward algorithm (GB)}
The greedy backward algorithm works in the opposite direction from the greedy forward algorithm. It begins with a complete measurement set. Sequentially, it takes out the measurement times one at a time such that, at each iteration, the dropped measurement time minimizes the cost. It stops when the set contains $N$ measurement times.


\subsubsection{Simulated annealing algorithm (SA)}

The implementation of the SA algorithm is done according to \cite[Section 1.2.3]{delahaye2019simulated}. The initial temperature parameter is $T\si{\degree}=10$ and the geometric temperature decay rate is $0.9$, i.e., after each generation, $T\si{\degree}\coloneqq 0.9\cdot T\si{\degree}$. If a mutation reduces the cost, it is kept. If it increases the cost, it is kept with probability $\exp(-(\text{cost\ increase})/T\si{\degree})$. The mutations are executed by prohibiting the duplication of measurement times and the mutation probability is $0.1$ by measurement time.

\subsubsection{Genetic algorithm (GA)}
\label{sec:Genetic algorithm (GA)}
The implementation of the genetic algorithm is based on \cite{mitchell1998introduction}. The selection of the individuals kept for the next generation is done according to stochastic universal sampling \cite[Section 5.4]{mitchell1998introduction}. To ensure that the number of measurement times remains constant over generations, a count preserving crossover operator is implemented \cite[Section 3.6]{umbarkar2015crossover}. It prevents the production of individuals $(t_{j+1},\dots,t_N)$ for which $t_k=t_l$ with $k\neq l$. For the same reason, the mutations are executed by prohibiting the duplication of measurement times. The algorithm uses sigma scaling with a unitary sigma coefficient \cite[Section 5.4]{mitchell1998introduction}. Crossover probability is 1 and mutation probability per measurement time is 0.003.

When the RT, SA or GA algorithm is used to solve the $N$ online Problems \ref{prob:online}, the solution found for the $j^\text{th}$ problem can be added to the initial population when solving the $(j+1)^\text{th}$ problem. This can speed up the convergence of the algorithm but also reduce the exploration capacity of these methods. In this work, we do not use this acceleration trick.

As mentioned in Subsection \ref{sec:numericalApprox}, the objective functions of Problems \ref{prob:offline} and \ref{prob:online} can be evaluated approximately only. This makes minimization difficult if it is too sensitive to bad cost function estimations. To face this issue, the SA algorithm and the GA return the best individual of the last generation instead of the best individual in all generations.

In Section \ref{sec:resultsAndDiscussion}, the performances of these five algorithms are compared. It will show that the GA outperforms the other algorithms on the studied example.



\section{Results and discussion} \label{sec:resultsAndDiscussion}

\subsection{Results parameters} \label{sec:resultsParam}
In this section, first we present two dynamical systems used for the results, then we describe the parameters of the simulations and the performance indicators.

\subsubsection{Tumor motion model} \label{sec:tumorModel}
To illustrate the performances of our method, we study a model describing one-dimensional tumor motion. As mentioned in introduction, the problem of tracking tumors with X-rays requires doing the best of each X-ray acquisition in order to spare healthy tissues.

The duration between each discrete time step is $\delta$. The tumor position to estimate $z(t)$ is a shifted sinusoidal signal with a time-varying amplitude $a(t)$, a time-varying shift $b(t)$, and a constant frequency $\omega$. Both $a(t)$ and $b(t)$ are bounded random walks. Bounds ensure that values remain realistic over time. In addition, the constant oscillation frequency $\omega$ is picked uniformly at random at the beginning of the process. Each measurement $y(t)$ is a noisy version of the position $z(t)$.

Let us define the state $x(t)=\begin{bmatrix} a(t) & b(t) & \omega(t) \end{bmatrix}^\top\in\R^3$ and the process noise $w(t)=\begin{bmatrix} w_a(t) & w_b(t) \end{bmatrix}^\top\in\R^2$. The dynamic of the system is defined as follows:
\begin{align}
x(t+1)&=\begin{bmatrix} a(t+1) \\ b(t+1) \\ \omega(t+1) \end{bmatrix}\nonumber\\
&=f(x(t),w(t)) \label{eq:tumorModel:x}\\
&\coloneqq\begin{bmatrix}
\text{clip}(a(t)+w_a(t),\underline{a},\bar{a})\\
\text{clip}(b(t)+w_b(t),\underline{b},\bar{b})\\
\omega(t)
\end{bmatrix},\nonumber\\
y(t)&=g_t(x(t),v(t))=a(t)\sin(\omega(t) t\delta)+b(t)+v(t), \label{eq:tumorModel:y}\\
z(t)&=h_t(x(t)) = a(t)\sin(\omega(t) t\delta)+b(t), \label{eq:tumorModel:z}\\
x(0)&=\begin{bmatrix} a(0) \\ b(0) \\ \omega(0) \end{bmatrix} \sim \mathcal{U} \left( [\underline{a},\bar{a}] \times [\underline{b},\bar{b}] \times [\underline{\omega},\bar{\omega}] \right).\label{eq:tumorModel:x0}
\end{align}
Equations (\ref{eq:tumorModel:x}), (\ref{eq:tumorModel:y}) and (\ref{eq:tumorModel:z}) hold for $t=0,\dots,T-1$, $t\in\{t_1,\dots,t_N\}$, and $t=0,\dots,T$, respectively. Noises $w_a(t)$, $w_b(t)$ and $v(t)$ are independent zero mean Gaussian noises and have a unitary standard deviation, i.e., $\sigma_a=\sigma_b=\sigma_v=1\ [mm]$. The clipping function is defined as $\text{clip}(x,\underline{x},\bar{x})\coloneqq\min(\max(x,\underline{x}),\bar{x})$. Equation (\ref{eq:tumorModel:x0}) indicates that the initial state is uniformly distributed at random on the indicated domain.

We consider $N=11$ measurements in the range of $0$ to $T=30$. To make this model more realistic, the values of $\underline{a}$, $\bar{a}$, $\underline{b}$, and $\bar{b}$ are inspired from \cite[Table 1]{dhont2018long} and the values of $\underline{\omega}$ and $\bar{\omega}$ are inspired from \cite[Table 1]{wuyts2011sigh}. The parameters of the system are
\begin{align*}
&\underline{a}=8.8\ [mm],\ \bar{a}=24\ [mm], \\
&\underline{b}=-5.8\ [mm],\ \bar{b}=5.8\ [mm],\\
&\underline{\omega}=1.3\ [rad/s], \ \bar{\omega}=2.1\ [rad/s],\\
&\delta=0.25\ [s],\ \sigma_a=\sigma_b=\sigma_v=1\ [mm],\\
&T=30,\ N=11.
\end{align*}

The importance density used by the particle filter is the prior density, except for $\omega(t)$ where we use $\omega(t+1)=\text{clip}(\omega(t)+w_\omega(t),\underline{\omega},\bar{\omega})$ with $w_\omega(t)$, an independent zero mean Gaussian noise with standard deviation $\sigma_\omega=0.005$.

%

\subsubsection{A common benchmark for particle filters}\label{sec:toy_ex}

We also test our method on the following widely studied model \cite{arulampalam2002tutorial, aspeel2020optimal,carlin1992monte,kitagawa1996monte,kadirkamanathan2002particle}:
\begin{align}
x(t+1) &=\ \dfrac{x(t)}{2} + \frac{25 x(t)}{1+x(t)^2} + 8 \cos(1.2 t) + w(t),  \label{eq:toy_ex:x}\\
y(t)   &=\ \frac{x(t)^2}{20} + v(t), \label{eq:toy_ex:y}\\
z(t)   &=\ x(t), \label{eq:toy_ex:z}\\
x(0) &\sim\ \N(0,5^2), \label{eq:toy_ex:x0}
\end{align}
where $x(t),y(t),z(t),w(t),v(t)\in\R$. Equations \eqref{eq:toy_ex:x}, \eqref{eq:toy_ex:y} and \eqref{eq:toy_ex:z} hold for $t=0,\dots,T-1$, $t\in\{t_1,\dots,t_N\}$, and $t=0,\dots,T$, respectively. Quantities $w(t)$, $v(t)$ and $x(0)$ are randomly distributed according to independent zero mean Gaussian distributions with standard deviations $\sigma_w=1$, $\sigma_v(t)=\sin(0.25t)+2$ and $\sigma_{x_0}=5$, respectively. As for the tumor model, $N=11$ measurements are allowed and the horizon is set at $T=30$. For this model, the importance density used by the particle filter is the prior density.

\subsubsection{Performance indicators} \label{sec:perfIndicators}
To test the proposed methods, we simulate many realizations of $\{y(t)\}_{t=0,\dots,T}$ and $\{z(t)\}_{t=0,\dots,T}$ according to system dynamic on which our filtering method is applied. Then, the filtering mean square error, $\text{MSE}=\frac{1}{T+1}\sum_{t=0}^T\|z(t)-\zh(t)\|^2$, is computed and is compared with the filtering mean square error $\text{MSE}_\text{REG}$ obtained with regular measurement times,
\begin{align}\label{eq:regularMeasurements}
t_{1:N}=\left\{\left. \text{Round}\left[\dfrac{kT}{N-1}\right] \right| k=0,\dots,N-1  \right\},
\end{align}
where $\text{Round}[\cdot]$ is the rounding operator. More precisely, we look at the gain
\begin{align}\label{eq:gain}
G\coloneqq \log_{10}\left( \dfrac{\text{MSE}_{\text{REG}} }{ \text{MSE} } \right).
\end{align}
It illustrates the merits of using optimal measurement times instead of regular ones. A positive gain indicates that the considered intermittent particle filter outperforms the regular one.

As performance indicators, we look at the mean and median gain and the proportion of positive gain over all the simulations. Notice that the mean gains can be seen as an estimation of the expected gain and the proportion of positive gains as an estimation of the probability of outperforming the regular measurement filter.


\subsubsection{Optimization and simulations parameters} \label{sec:optiAndSimParam}


As the measurement times that are the solution of the online Problem \ref{prob:online} are a function of the previous measurements, the optimization algorithm has to be run on each simulation. This increases the computation time significantly. Therefore, the solutions of Problems \ref{prob:offline} and \ref{prob:online} are studied separately and the number of simulations and the optimization parameters have been set accordingly. The values of the optimization parameters and the simulation parameters are given in Table \ref{tab:parametersValues}. Note that the number of particles used for the optimization is different from the number of particles for the filtering.
\begin{table}[ht]
\centering
\begin{tabular}{|l||l|l|l|}
\hline
               & Parameters & Offline  & Online  \\ \hline\hline
\multirow{4}{*}{Optimization} & \# draws $K$       & 1000                 & 200                \\ \cline{2-4}
& \# particles       & 200                 & 100                \\ \cline{2-4}
&pop. size      & 50                 & 30                 \\ \cline{2-4}
&\# generations      & 25                  & 15                 \\ \hline\hline
\multirow{2}{*}{Simulations} & \# simulations       & 100000              & 500                \\ \cline{2-4}
& \# particles & 1000                 & 1000                \\ \cline{2-4}
\hline
\end{tabular}%
\caption{Values of the model parameters, test parameters, and optimization parameters for the offline Problem \ref{prob:offline} and the online Problem \ref{prob:online}.}
\label{tab:parametersValues}
\end{table}

\subsection{Comparison of optimization algorithms} \label{sec:compareOptimizationAlgo}

The different optimization algorithms proposed in Subsection \ref{sec:optimizationAlgo} are tested on the offline Problem \ref{prob:offline} for the tumor motion model described in Section \ref{sec:tumorModel}.

For all the optimization algorithms, the evolution of the cost of Problem \ref{prob:offline} with respect to the number of cost function evaluations is illustrated in Figure \ref{fig:algoComparison}. Because the number of cost function evaluations of the GF and GB algorithm is fixed, they are represented by points in the figure.

\begin{figure}[ht]
    \centering
    \includegraphics[width=8.5cm]{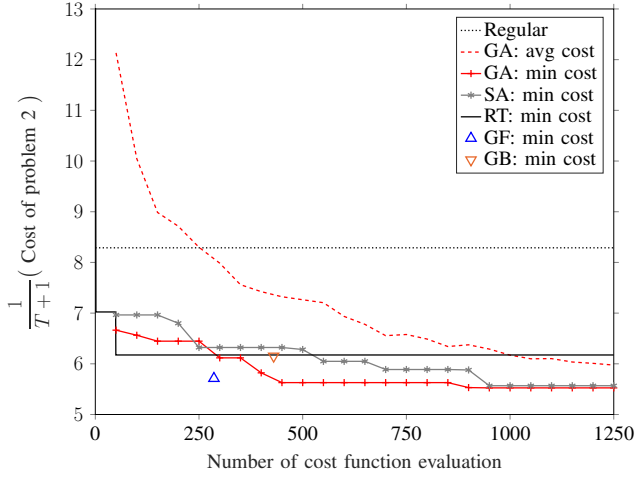}
    \caption{Evolution of the minimum cost of the offline Problem \ref{prob:offline} with respect to the number of cost function evaluations for the different optimization algorithms: genetic algorithm (GA), simulated annealing (SA), random trial (RT), greedy forward (GF), and greedy backward (GB). It illustrates the quality of minimization for given computational resources. The evolution of the average cost of the population of the GA algorithm at each generation and the cost of regularly spaced measurement times are shown as well.}
    \label{fig:algoComparison}
\end{figure}

One can observe that all five optimization algorithms return a measurement subset with associated expected MSE significantly lower than the regularly spaced measurements. The final expected MSE of the SA, RT, GF and GB optimization algorithms all lie close to $6$. Note that the GF and GB algorithms give good results for a relatively small number of evaluations of the cost function. Consequently, they can be a good option if computing resources are limited. The GA and SA algorithms perform notably better with a final minimum cost around $5.5$. One can note that the difference between GA average and minimum cost decreases over generations before reaching quasi-convergence, meaning that most of the individuals are identical, which is the desired convergence behavior for a GA.

As a result of these observations, the GA is used as the optimization algorithm in the rest of the paper.

\subsection{Filtering performance}\label{sec:filteringPerformance}
In this subsection, we analyze the filtering performance obtained using the offline (Problem \ref{prob:offline}) and online (Problem \ref{prob:online}) measurement times. First, we report the results on the tumor motion model (Subsection \ref{sec:tumorModel}), then we present the results for the common benchmark (Subsection \ref{sec:toy_ex}).

\subsubsection{Results for the tumor motion model (Subsection \ref{sec:tumorModel})}
The histogram of the gains $G$ for the offline Problem \ref{prob:offline} and for the online Problem \ref{prob:online} can be found in Figure \ref{fig:histo_both:tumor}. The offline gain is computed over 100000 simulations and the online gain is computed over 500 simulations. The vertical line corresponds to the null gain. These histograms can be interpreted as approximations of the probability density of the gains $G$. The proportion of positive gain for the GA algorithm can be read as the area of the histogram above $0$. The corresponding performance indicators are reported in Table \ref{tab:performanceIndicators:tumor} for both the offline particle filter (related to Problem \ref{prob:offline}) and online particle filter (related to Problems \ref{prob:online}).

The mean and median gains are positive for our offline and online methods, which indicates that they generally outperform the regular particle filter. In addition, the proportion of positive gain indicates that our offline and online methods outperform the regular particle filter in \AA{$69.3\%$} and \AA{$76.3\%$} of the cases, respectively.




\begin{figure}[ht]
    \centering
    \includegraphics[width=8.5cm]{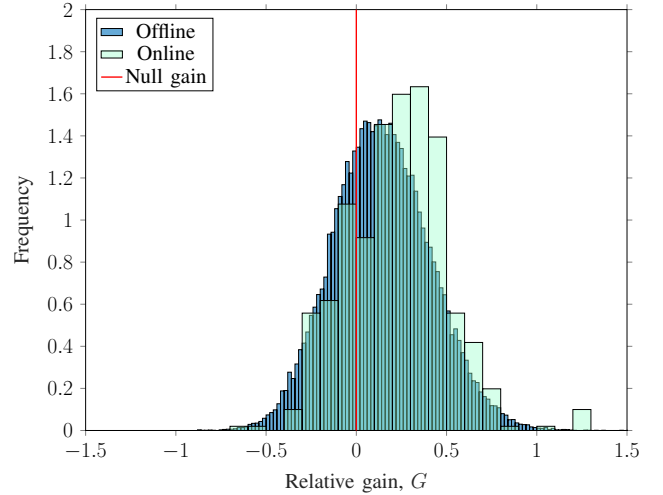}
    \caption{Histograms of the gain $G$ (see Equation (\ref{eq:gain})) for the tumor motion model (Subsection \ref{sec:tumorModel}). The histograms have been obtained over 100000 and $500$ simulations for the offline and online particle filters, respectively.}
    \label{fig:histo_both:tumor}
\end{figure}

\begin{table}[ht]
\centering
\begin{tabular}{|l||l|l|}
\hline
                         & Offline      & Online   \\ \hline\hline
Mean gain ($\pm$ std)    & 0.142 ($\pm$ 0.269) &  0.238 ($\pm$ 0.296) \\ \hline
Median gain              & 0.134      & 0.248   \\ \hline
Proportion positive gain & 69.3 \%     & 76.3 \%  \\ \hline
\end{tabular}%
\caption{Performance indicators for the offline and online solutions. The considered system, the performance indicators, and the optimization and simulation parameters are described in Subsections \ref{sec:tumorModel}, \ref{sec:perfIndicators}, and \ref{sec:optiAndSimParam}, respectively.}
\label{tab:performanceIndicators:tumor}
\end{table}



Figure \ref{fig:tumorModel:draw} represents a particular trajectory $ z (t) $ and the trajectory filtered by a particle filter with regular measurement times (see Equation (\ref{eq:regularMeasurements})), as well as with the offline and online particle filters. In all three cases, the measurement times are indicated. As expected, we obtain better filtering performance using the online particle filter (the gain is $G_{\text{online}}=0.20$) than the offline particle filter (the gain is $G_{\text{offline}}=0.15$), the largest mean square filtering error being obtained with the regular measurements.

\begin{figure}[ht]
    \centering
    \includegraphics[width=8.5cm]{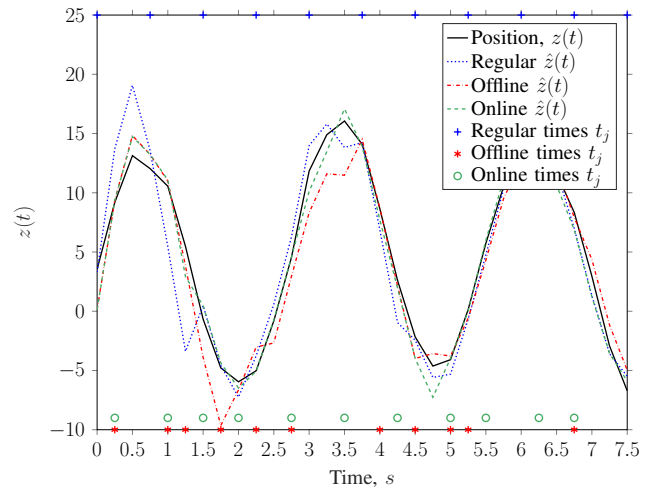}
    \caption{Comparison of a particular realization $z(t)$ with the filtered values $\zh(t)$ obtained with a particle filter with regular measurement times and both the offline and the online particle filters. The gains on the complete sequence are $G_{\text{offline}}=0.223$ and $G_{\text{online}}=0.616$. Results are simulated from the tumor motion model (Subsection \ref{sec:tumorModel}).}
    \label{fig:tumorModel:draw}
\end{figure}

\subsubsection{Results for the common benchmark (Subsection \ref{sec:toy_ex})}

Figure \ref{fig:histo_both:toy} represents the histogram of the offline and online gains $G$ for the common benchmark. The corresponding performance indicators are reported in Table \ref{tab:performanceIndicators:toy}. It is observed that both the offline and online methods outperform the particle filter with regular measurement times. In addition, the performance difference between the offline and online methods is important. The fact that the performance gap between the offline and online particle filters is larger for this common benchmark than for the tumor motion model can be related to Remark \ref{remark:kalman}. Indeed, if the model were linear and Gaussian, there would be no performance gap.

\begin{figure}[ht]
    \centering
    \includegraphics[width=8.5cm]{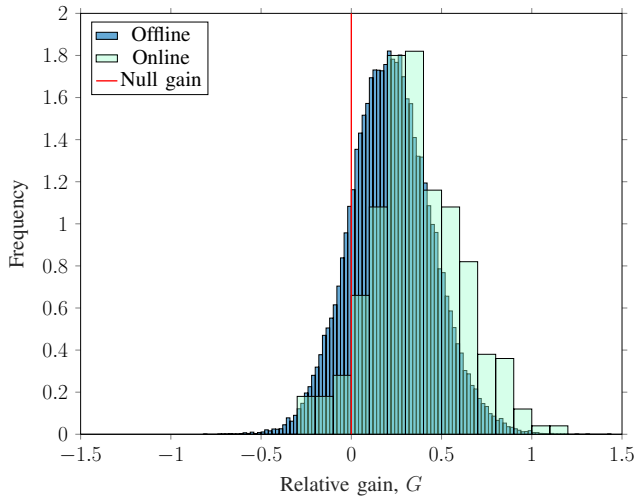}
    \caption{Histograms of the gain $G$ (see Equation (\ref{eq:gain})) for the common benchmark (Subsection \ref{sec:toy_ex}). The histograms have been obtained over 100000 and $500$ simulations for the offline and online particle filters, respectively.}
    \label{fig:histo_both:toy}
\end{figure}

\begin{table}[ht]
\centering
\begin{tabular}{|l||l|l|}
\hline
                         & Offline      & Online   \\ \hline\hline
Mean gain ($\pm$ std)    & 0.221 ($\pm$ 0.225) &  0.365 ($\pm$ 0.253) \\ \hline
Median gain              & 0.216     & 0.340   \\ \hline
Proportion positive gain & 84.3 \%     & 93.8 \%  \\ \hline
\end{tabular}%
\caption{Performance indicators for the offline and online solutions. The considered system, the performance indicators, and the optimization and simulation parameters are described in Subsections \ref{sec:toy_ex}, \ref{sec:perfIndicators}, and \ref{sec:optiAndSimParam}, respectively.}
\label{tab:performanceIndicators:toy}
\end{table}

Finally, Figure \ref{fig:toy:draw} shows a trajectory of the common benchmark as well as the tracking obtained by the regular, offline and online particle filters. The measurement times are also displayed.

\begin{figure}[ht]
    \centering
    \includegraphics[width=8.5cm]{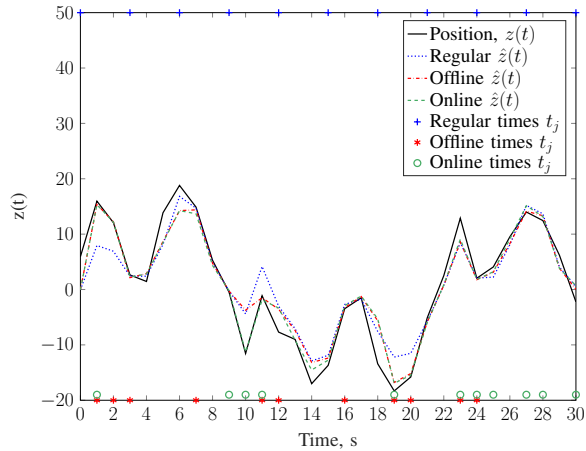}
     \caption{Comparison of a particular realization $z(t)$ with the filtered values $\zh(t)$ obtained with a particle filter with regular measurement times and both the offline and the online particle filters. The gains on the complete sequence are $G_{\text{offline}}=0.173$ and $G_{\text{online}}=0.283$. Results are simulated from the common benchmark (Subsection \ref{sec:toy_ex}).}
    \label{fig:toy:draw}
\end{figure}

\section{Conclusion}\label{sec:conclusion}
The problem of the optimal intermittent particle filter has been addressed. This consists in selecting the measurement times of the particle filter according to a certain optimality criterion.

We have proposed three variants of expected mean square error minimization giving rise to three different intermittent filters: the stochastic program filter, the offline filter, and the online filter.

A theorem has been proven that compares the performance of these three intermittent filters.

Then, different algorithms to compute the measurement times of the offline and online particle filters have been proposed: the random trial, greedy forward, greedy backward, simulated annealing, and genetic algorithms. These different optimization algorithms were compared on a tumor movement model inspired by real data. The genetic and simulated annealing algorithms are the two optimizers showing the best performance.

Finally, on this same tumor model, our offline and online particle filters were compared with a particle filter with regular measurement times. The offline and online particle filters outperform the regular particle filter in \AA{$69.3\%$} and \AA{$84.3\%$} of the cases, respectively. For a second example, the regular particle filter is outperformed in \AA{$76.3\%$} and \AA{$93.8\%$} of the cases by the offline and online particle filters, respectively.

Further work will focus on robustness analyses to determine how model uncertainties affect the performance of the proposed methods for situations in which the system dynamics are known only approximately. In addition, we will use a reinforcement learning approach to solve the stochastic program and test the stochastic program particle filter. Other continuations of this work will consist of experimental validation on real-world data, for example, in the case of mobile tumor tracking. It would also be interesting to investigate a continuous time version of the problem.

Overall, our results demonstrate the added value of selecting measurement times according to system dynamics for the optimal estimation of a state from limited measurements.




\appendix

\section*{Proof of Theorem \ref{thm:SP<online<offline}}
The structure of the proof consists in (i) observing that once all the measurements are known, cost-to-go is equal for the three filters; then (ii) we show that if the cost-to-go when choosing the measurement time $ t_ {j + 1} $ is smaller with the stochastic program filter than with the online filter, which itself is smaller than with the offline filter, then this is also the case when choosing the measurement time $ t_j $; and finally (iii) a backward recursive argument shows that this is also true for the initial cost-to-go.

This structure is applied to each inequality separately.

\paragraph*{Proof of the first inequality}
From Relations (\ref{eq:SP:terminalCond}) and (\ref{eq:online:terminalCond}) it follows that for all $t_{1:N}$ and $y(t_{1:N})$,
\begin{align}\label{eq:proof1:terminalCond}
V_{N}(t_{1:N};y(t_{1:N}))=F_{N}(t_{1:N};y(t_{1:N})).
\end{align}

For a fixed $j\in\{0,\dots,N-1\}$, assume that for all $t_{1:j+1}$ and $y(t_{1:j+1})$ it holds that
\begin{align}\label{eq:proof1:j+1}
V_{j+1}(t_{1:j+1};y(t_{1:j+1})) \leq F_{j+1}(t_{1:j+1};y(t_{1:j+1})).
\end{align}
Then, taking the expectation with respect to $y(t_{j+1})$ and adding a same term on both sides, we have,
\begin{align*}
&\sum_{t=t_j}^{t_{j+1}-1}\Ec[t|y(t_{1:j})] + \E_{y(t_{j+1})}\Bigg[ V_{j+1}(t_{1:j+1};y(t_{1:j+1})) \Bigg] \leq\\
&\sum_{t=t_j}^{t_{j+1}-1}\Ec[t|y(t_{1:j})] + \E_{y(t_{j+1})}\Bigg[ F_{j+1}(t_{1:j+1};y(t_{1:j+1})) \Bigg],
\end{align*}
for all $t_{1:j+1}$ and $y(t_{1:j})$. Taking the minimum over all $t_{j+1}$ on the left-hand side and fixing $t_{j+1}$ to $t_{j+1}^o$ on the right-hand side, we have,
\begin{align*}
&\min_{t_{j+1}}\Bigg\{\sum_{t=t_j}^{t_{j+1}-1}\Ec[t|y(t_{1:j})] \\
&+ \E_{y(t_{j+1})}\Bigg[ V_{j+1}(t_{1:j+1};y(t_{1:j+1})) \Bigg]\Bigg\} \leq\\
&\sum_{t=t_j}^{t_{j+1}^o-1}\Ec[t|y(t_{1:j}))] \\
&+ \E_{y(t_{j+1}^o}) \Bigg[ F_{j+1}(t_{1:j},t_{j+1}^o;y(t_{1:j}),y(t_{j+1}^o)) \Bigg],
\end{align*}
for all $t_{1:j}$ and $y(t_{1:j})$. From Relation (\ref{eq:SP:recursion}) on the left-hand side and Relation (\ref{eq:online:recursion}) on the right-hand side, we have
\begin{align}\label{eq:proof1:j}
V_{j}(t_{1:j};y(t_{1:j})) \leq F_{j}(t_{1:j};y(t_{1:j})),
\end{align}
for all $t_{1:j}$ and $y(t_{1:j})$.

But then, because Relation (\ref{eq:proof1:j+1}) implies (\ref{eq:proof1:j}) and because of (\ref{eq:proof1:terminalCond}), we have $V_0(\emptyset;\emptyset) \leq F_0(\emptyset;\emptyset)$.

\paragraph*{Proof of the second inequality}
It follows from a similar argument. We define
\begin{align}\label{eq:proof2:terminalCond}
J_{N}(t_{1:N};y(t_{1:N})):=F_{N}(t_{1:N};y(t_{1:N})).
\end{align}
Then, for a fixed $j\in\{0,\dots,N-1\}$, assume that for all $t_{1:j+1}$ and $y(t_{1:j+1})$ it holds that
\begin{align}\label{eq:proof2:j+1}
F_{j+1}(t_{1:j+1};y(t_{1:j+1})) \leq J_{j+1}(t_{1:j+1};y(t_{1:j+1})).
\end{align}
Rewriting the right-hand side according to (\ref{eq:def:J}) gives
\begin{align*}
&F_{j+1}(t_{1:j+1};y(t_{1:j+1}))\leq \\
&\min_{t_{j+2}<\cdots<t_N\leq T} \E_{y(t_{j+2:N})}\Bigg[ \sum_{t=t_{j+1}}^{T}\Ec[t|y(t_{1:N})] \Bigg].
\end{align*}
The right-hand side is upper bounded by
\begin{align*}
\E_{y(t_{j+2:N})}\Bigg[ \sum_{t=t_{j+1}}^{T}\Ec[t|y(t_{1:N})] \Bigg],
\end{align*}
for all $t_{j+2:N}$. Then one can write
\begin{align*}
&F_{j+1}(t_{1:j+1};y(t_{1:j+1}))\leq \\
&\E_{y(t_{j+2:N})}\Bigg[ \sum_{t=t_{j+1}}^{T}\Ec[t|y(t_{1:N})] \Bigg],
\end{align*}
for all $t_{1:N}$ and $y(t_{1:j+1})$. Taking the expectation with respect to $y(t_{j+1})$ and adding a same term on both sides gives
\begin{align}
&\sum_{t=t_j}^{t_{j+1}-1}\Ec[t|y(t_{1:j})] + \E_{y(t_{j+1})}\Bigg[ F_{j+1}(t_{1:j+1};y(t_{1:j+1})) \Bigg]\leq \nonumber\\
&\sum_{t=t_j}^{t_{j+1}-1}\Ec[t|y(t_{1:j})] +\E_{y(t_{j+1:N})}\Bigg[ \sum_{t=t_{j+1}}^{T}\Ec[t|y(t_{1:N})] \Bigg],\label{eq:proof:auxIneq}
\end{align}
for all $t_{1:N}$ and $y(t_{1:j})$. Using the fact that the first term of the right-hand side is independent of $y(t_{j+1:N})$, and then thanks to Remark \ref{remark:notation}, the right-hand side can be rewritten successively as
\begin{align*}
&\E_{y(t_{j+1:N})}\Bigg[\sum_{t=t_j}^{t_{j+1}-1}\Ec[t|y(t_{1:j})] + \sum_{t=t_{j+1}}^{T}\Ec[t|y(t_{1:N})] \Bigg]= \\
&\E_{y(t_{j+1:N})}\Bigg[\sum_{t=t_j}^{t_{j+1}-1}\Ec[t|y(t_{1:N})] +\sum_{t=t_{j+1}}^T \Ec[t|y(t_{1:N})]\Bigg].
\end{align*}
By combining the two sums on this right-hand side, the inequality (\ref{eq:proof:auxIneq}) becomes
\begin{align*}
&\sum_{t=t_j}^{t_{j+1}-1}\Ec[t|y(t_{1:j})] + \E_{y(t_{j+1})}\Bigg[ F_{j+1}(t_{1:j+1};y(t_{1:j+1})) \Bigg]\leq \\
&\E_{y(t_{j+1:N})}\Bigg[ \sum_{t=t_{j}}^{T}\Ec[t|y(t_{1:N})] \Bigg],
\end{align*}
for all $t_{1:N}$ and $y(t_{1:j})$. We set $t_{j+1:N}$ to the value that minimizes this right-hand side. In so doing, $t_{j+1}=t_{j+1}^o$ and the inequality becomes
\begin{align*}
&\sum_{t=t_j}^{t_{j+1}^o-1}\Ec[t|y(t_{1:j})] \\
&+ \E_{y(t_{j+1}^o)}\Bigg[ F_{j+1}(t_{1:j},t_{j+1}^o;y(t_{1:j}),y(t_{j+1}^o)) \Bigg]\leq \\
&\min_{t_{j+1}<\cdots<t_N\leq T}\E_{y(t_{j+1:N})}\Bigg[ \sum_{t=t_{j}}^{T}\Ec[t|y(t_{1:N})] \Bigg],
\end{align*}
for all $t_{1:j}$ and $y(t_{1:j})$. Now, using Relation (\ref{eq:online:recursion}) on the left-hand side and Relation (\ref{eq:def:J}) on the right-hand side, we obtain
\begin{align}\label{eq:proof2:j}
F_{j}(t_{1:j};y(t_{1:j})) \leq J_{j}(t_{1:j};y(t_{1:j})),
\end{align}
for all $t_{1:j}$ and $y(t_{1:j})$.

Finally, because Relation (\ref{eq:proof2:j+1}) implies Relation (\ref{eq:proof2:j}) and because of (\ref{eq:proof2:terminalCond}), we have $F_0(\emptyset;\emptyset)\leq J_0(\emptyset;\emptyset)$, which concludes the proof.

\bibliographystyle{IEEEtran}
\bibliography{refs}
\end{document}